\documentclass[11pt,a4paper]{article}

\usepackage {amsmath, amssymb, amsfonts,bm}
\usepackage {graphics}
\usepackage {graphicx}
\usepackage{natbib,url}

\usepackage {longtable}
\usepackage {float,afterpage,color} 
\usepackage {verbatim} 
\usepackage {subfigure}
\usepackage{chemarr} 
\usepackage{lscape}
\usepackage{bbm,dsfont}
\usepackage{setspace}
\usepackage{fancyheadings}
\usepackage{float,afterpage}
\usepackage{multirow}
\usepackage{amsthm}
\usepackage{color}

\usepackage{algorithm}
\usepackage{algorithmicx,algpseudocode}

\usepackage{anysize}
\marginsize{2cm}{2cm}{2cm}{2cm}

\DeclareMathOperator{\argmax}{argmax}
\DeclareMathOperator{\argmin}{argmin}

\newcommand{\St}{\mathcal{S}}
\newcommand{\U}{\mathcal{U}}
\newcommand{\T}{\mathcal{T}}
\newtheorem{theorem}{Result}

\newcommand{\ie}{i.e.\ }
\newcommand{\eg}{e.g.\ }

\newcommand{\pr}{\text{Pr}}

\newcommand{\be}{\begin{equation}}
\newcommand{\ee}{\end{equation}}

\title{Considerate Approaches to Achieving Sufficiency for ABC model selection}
\author{Chris Barnes$^{1,\ast}$, Sarah Filippi$^{1,\ast}$, Michael P.H. Stumpf$^{1,\ast,\#}$, Thomas Thorne$^{1,\ast}$ \\ \small $^1$Centre for Integrative Systems Biology and Bioinformatics,\\ \small Imperial College London, London SW7 2AZ, UK. }

\begin{document}

\maketitle
\begin{abstract}
For nearly any challenging scientific problem evaluation of the likelihood is problematic if not impossible. Approximate Bayesian computation (ABC) allows us to employ the whole Bayesian formalism to problems where we can use simulations from a model, but cannot evaluate the likelihood directly. When summary statistics of real and simulated data are compared --- rather than the data directly --- information is lost, unless the summary statistics are sufficient. Here we employ an information-theoretical framework that can be used to construct (approximately) sufficient statistics by combining different statistics until the loss of information is minimized. Such sufficient sets of statistics are constructed for both parameter estimation and model selection problems. We apply our approach to a range of illustrative and real-world model selection problems.
\end{abstract}

\noindent
\small $^\ast$ All authors contributed equally.\\
\noindent
\small $^\#$ To Whom Correspondence Should be Addressed: m.stumpf@imperial.ac.uk
\section{Introduction}
Mathematical models are widely used to describe and analyze complex systems and processes across the natural, engineering and social sciences.   Formulating a model to describe, e.g. a predator-prey system, geophysical process, communication system, or social network requires us to condense our assumptions and knowledge into a single coherent framework \citep{May:2004p11952}. In constructing these models we have to state our assumptions about their constituent parts and their interactions explicitly. Mathematical analysis or computer simulations of these models then allows us to compare model predictions with experimental observations in order to test and ultimately improve the models. Even previously largely observational sciences such as biology, geology and meteorology are now heavily influenced by computer simulations which are employed for explanatory as well as predictive purposes. 
\par
Because many of the mathematical models in these disciplines are too complicated to be analyzed in closed form, computer simulations have become the primary tool in the theoretical analysis of very large or complex models. Modelling such systems is often (relatively) straightforward  if the mathematical structure of the model and reliable estimates for the model parameters are known. Unfortunately, it is considerably harder  to infer the structure of mathematical models and estimate their respective parameters based on experimental data, in particular if we seek to infer a model that can describe the data. Whenever probabilistic models exist we can employ standard model selection approaches of either a frequentist, Bayesian, or information theoretic nature \citep{Cox1974,Mackay:2003aa,Burnham:2002p4089}. But if suitable probability models do not exist, or if the evaluation of the likelihood is computationally intractable, then we have to base our assessment on the level of agreement between simulated and observed data. This is particularly challenging when the parameters of simulation models are not known but must be inferred from observed data as well. 
\par
 For such cases --- cases where conventional statistical approaches fail because of the enormous computational burden incurred in evaluating the likelihood --- so-called approximate Bayesian computation (ABC) schemes have recently come to the fore \citep{Pritchard:1999td,Beaumont:2002ue,Tanaka:2006fj,Secrier:2009ko}. These forgo the explicit evaluation of the likelihood by a principled comparison between the observed and simulated data. In many cases inferences are furthermore based not on the data themselves, but on summary statistics of the data. Such statistics serve as data compression tools\citep{Cover:2006aa} and, if used sensibly, enable computationally efficient inference from data sets, where the complexity of the data would stymie conventional likelihood-based methods \citep{Pritchard:1999td,Ratmann:2007hh}.
 \par
 ABC schemes have become increasingly popular, because of their flexibility and their deceptive conceptual simplicity. While especially some computationally demanding areas have fuelled the development of powerful ABC approaches, notably population genetics \citep{Fagundes:2007fr}, evolutionary biology \citep{Wilkinson:2010fc}, systems biology \citep{Liepe:2010eg}, dynamical systems theory \citep{Toni:2009gm}, and epidemiology \citep{Blum:2008cr}, a worrying increase in naive (and plainly incorrect, see \eg \cite{Walker:2010ki}) applications are beginning to emerge. Such problems, as recent results by \cite{Didelot:2010wo} and \cite{Robert:2011wz} suggest, are more imminent in model selection applications. In the present context, all problems stem back to the issue of {\em sufficiency} of statistics, and its role in model selection. The present paper sets out to develop remedies for such problems.  Below we will begin with an outline of the basic ideas underlying ABC, before discussing the particular challenges raised in particular by Robert and colleagues. We will then list the cases where  ABC-based model selection is possible; in essence it is the ill-judged use of summary statistics and failure to ensure sufficiency which lies at the heart of the problem identified by \cite{Robert:2011wz}, before setting out methods that allow us to remedy these problems. We illustrate the use of these methods in a number of applications before concluding with some more general remarks on the conceptual and mathematical foundations of ABC approaches.

\section{Approximate Bayesian Computation}
\subsection{Sufficient Statistics}
Bayesian inference centres around the {\em posterior distribution},
\be
p(\theta|x) = \frac{f(x|\theta)\pi(\theta)}{p(x)}
\label{eq:posterior}
\ee
where $x$ are the data, which are drawn from some sample space, $x\in\Omega\subseteq \mathds{R}^D$, $f(x|\theta)$ is the {\em likelihood}, $\pi(\theta)$ the {\em prior distribution}, and $\theta$ an unknown parameter \citep{Robert:2007aa}; $p(x)$ is often called the evidence \citep{Mackay:2003aa}, but in many applications or discussions dismissed as a normalization constant. The {\em Likelihood principle} states that all the information about parameter $\theta$ is contained in the likelihood function $f(x|\theta)$, \ie once we have the form of the likelihood, we do not have to retain any of the data. This principle is complemented by the {\em sufficiency principle}. Here a summary statistic of the general form
\be
\St: \mathbb{R}^d \longrightarrow \mathbb{R}^w, \quad\St(x) = s 
\ee
with $w\ll d$ typically, is called sufficient if the likelihood is independent of the parameter conditional on the value of the summary statistic. We denote by $f(x|\theta,s)$ the likelihood conditionally on the value of the summary statistic $\St(x)=s$ and $g(x|s)$ the density probability of the data given the summary statistic. The statistic is sufficient if and only if:
$$
f(x|s,\theta) = g(x|s).
$$
The likelihood can then generally be written in the Neyman-Fisher factorized form
\be
f(x|\theta) = g(x|s)\tilde{f}(s|\theta),
\label{eq:factor}
\ee
where $s=\St(x)$ and $\tilde{f}(s|\theta)$ is the likelihood of the sufficient statistic \citep{Cox2006}. The function $g(\cdot)$ is independent of the parameter $\theta$. Thus $\tilde{f}(s|\theta)$ carries all the information about the parameter. 
\par
This factorization is, however, not unique, as it depends on the sufficient statistic, which is generally not unique. For example, any statistic containing additional information in addition to a sufficient statistic is also sufficient. Therefore we typically seek to determine the minimally sufficient statistic, which is generally unique. As a consequence the functional forms (and values) of $g(\cdot)$ and $\tilde{f}(\cdot|\theta)$ depend on the choice of sufficient statistics. 
\par
In order to understand the terms in the factorization theorem, Eqn. (\ref{eq:factor}), we focus on the case where $X$ is a discrete random variable, we then have 
$$
g(x|s) = \pr(X=x|\St(X)=s)
$$
and 
$$
\tilde{f}(s|\theta) = \pr(\St(X)=s|\theta).
$$
Thus $g(x|s)$ is really the conditional probability of $X$ given an observed value for the summary statistic, $\St(X)=s$ and it is therefore linked to the compression of the data achieved by the summary statistic. Here it is worth remembering that the complete data also form a valid summary statistic, and for this choice of statistic we have trivially, for all $s$, $g(x|s)=1, \ \ \forall x\in \Omega$.

\subsection{ABC for parameter inference}
In practical applications we are  interested in evaluating the posterior distribution for model parameters, $\theta$, defined Eqn. (\ref{eq:posterior}).
When the likelihood is hard to evaluate it is still often possible to simulate from the model according to $f(\cdot|\theta)$. It is easy to show that for simulated data $y$ we have 
\be
p(\theta|x) = \frac{\mathds{1}(x=y)f(y|\theta)\pi(\theta)}{p(x)},
\ee
which in many practical applications can be approximated using suitable distance functions, $\Delta(x,y)$, whence, after marginalization over simulated data we get,
\be
p(\theta|x) \approx \int_\Omega \frac{\mathds{1}(\Delta(x,y)\le \epsilon) f(y|\theta)\pi(\theta)}{p(x)} dy.
\label{eq:ABC1}
\ee
This is obviously correct, as $\epsilon\longrightarrow 0$. 
\par
Based on the fact that sufficient statistics contain all the information about the $\theta$ that is contained in the data, we may be tempted to replace the data by the corresponding summary statistics. We thus replace the comparison of the data in Eqn. (\ref{eq:ABC1}) by a comparison of the values of their respective summary statistics, using a distance function which, by abusing the notation, is denoted by $\Delta(\St(x),\St(y))$,
\begin{align}
p(\theta|x)& \approx  \int_\Omega \frac{\mathds{1}(\Delta(\St(x),\St(y))\le \epsilon) f(y|\theta)\pi(\theta)}{p(x)} dy\nonumber\\
&= \frac{\int_{\mathds{R}^D} \mathds{1}(\Delta(\St(x),s)\le \epsilon) \tilde{f}(s|\theta)\pi(\theta)ds}{\int_\Theta\int_{\mathds{R}^w} \mathds{1}(\Delta(\St(x),s)\le \epsilon) \tilde{f}(s|\theta)\pi(\theta)ds d\theta},
\label{eq:ABC2}
\end{align}
where we have made it explicit in the second line that once we use summary statistics we are only considering the second term on the right-hand side of Eqn. (\ref{eq:factor}); any dependence on $g(\cdot)$ is lost, and so is therefore the effect of the data-compression in the summary statistic. Something similar is also implicit in conventional Bayesian inference. \cite{Cox2006} reinforces this point by stating that
 ``{\em Any Bayesian inference uses the data only via the minimal sufficient statistic. This is because the calculation of the posterior distribution involves multiplying the likelihood by the prior and normalizing. Any factor of the likelihood that is a function of $y$ alone will disappear after normalization.}"
\par
Quite generally, the choice of the summary statistic is important: without sufficiency the whole inference will only map the parameter regimes that will lead to model behaviour which embodies the constraints implied by specified summary statistic. Only if the summary statistic is sufficient, however, will we be able to infer the model parameters \citep{Fearnhead:2010vj}. For some non-sufficient statistics, however, some aspects of the true posterior can be elucidated as shown in Figure \ref{fig:stat-param}. We will return to a discussion of non-sufficient statistics later.  

\subsection{ABC for model selection}
One of the perhaps most useful (and aesthetically pleasing) aspects of the Bayesian inferential frameworks is that model selection is natural and intrinsic, especially compared to frequentist frameworks \citep{Robert:2007aa}. From the earliest days ABC approaches for model selection have also been promoted, see \eg \cite{Beaumont:2002ue} and \cite{Fagundes:2007fr}. Recently, however, \cite{Robert:2011wz} have issued a note of caution. This is based on the observation that a statistic, or set of statistics, which is sufficient for model parameters in different models, may still not be sufficient across models \citep{Didelot:2010wo,Robert:2011wz}. 
\par
\begin{figure}[tph]
\begin{center}
\includegraphics[width=0.8\textwidth]{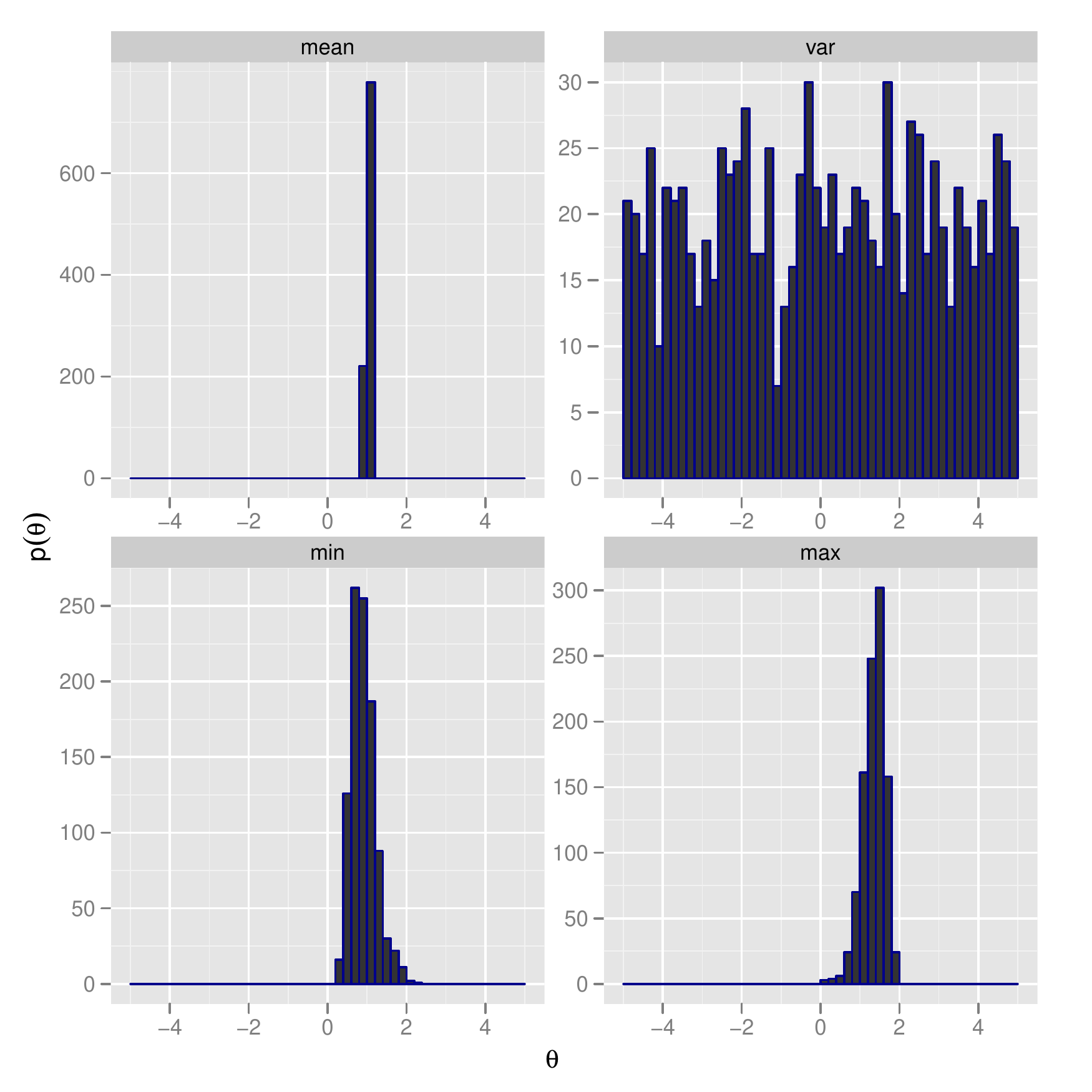}
\caption{Parameter inference for  the mean of a normal model with known standard deviation, $\sigma^2=1$, using the mean, variance, maximum and minimum as a statistic. We find that only the truly sufficient statistic ($\theta=\mu=1$) yields the correct posterior distribution despite the fact that we generated 1,000 acceptances of samples of size 10,000 with $\epsilon=0.001$.} 
\label{fig:stat-param}
\end{center}
\end{figure}
To illustrate this point we now consider a finite set of models, ${\cal M}=\{M_1,\ldots,M_q\}$, each of which has an associated parameter vector $\theta_m\in\Theta_m, \ 1\le m\le q$. We aim to perform inference on the {\em joint space} over models and parameters, $(m,\theta_m)$. \cite{Robert:2011wz} have focused on the Bayes Factors, but, of course, similar problems arise also for the marginal model likelihoods,
\be
p(m|x) =\frac{\int_{\Theta_m} f(x|\theta_m)\pi(\theta_m)d\theta_m \pi(m)}{\sum_{i=1}^q\int_{\Theta_i} f(x|\theta_i)\pi(\theta_i)d\theta_i\pi(i)}.
\ee
Again, we can apply ABC by replacing evaluation of the likelihood in favour of comparing simulated and real data for different parameters drawn from the posterior, whence we obtain
\be
p(m|x) \approx \frac{\int_{\Theta_m} \int_\Omega \mathds{1}(\Delta(x,y)\le \epsilon) f(y|\theta_m)\pi(\theta_m)d\theta_m dy\ \pi(m)}{\sum_{i=1}^q\int_{\Theta_i}  \int_\Omega \mathds{1}(\Delta(x,y)\le \epsilon) f(y|\theta_i)\pi(\theta_i)d\theta_i dy\ \pi(i)},
\ee
which is of course always exact once $\epsilon\longrightarrow 0$. The same is no longer true, however, once the complete data have been replaced by summary statistics. So in general
\be
p(m|x)\neq \frac{\int_{\Theta_m} \int_{\mathds{R}^w} \mathds{1}(\Delta(\St_m(x),s_m\le \epsilon) \tilde{f}(s_m|\theta_m)\pi(\theta_m)d\theta_m ds_m\ \pi(m)}{\sum_{i=1}^q\int_{\Theta_i}  \int_{\mathds{R}^w} \mathds{1}(\Delta(\St_i(x),s_i)\le \epsilon) \tilde{f}(s_i|\theta_i)\pi(\theta_i)d\theta_i ds_i\ \pi(i)};,
\ee
where $\St_i$, $1\leq i\leq q$ are the summary statistics for each model. An equality can only hold if the factors $g_i(x),\ 1\le i\le q$ are all identical. Otherwise the different levels of data-compression are lost and unbiased model selection is no longer possible.
\par
\subsection{Resuscitating ABC Model Selection}  
As shown by \cite{Robert:2011wz} and \cite{Didelot:2010wo} even if we choose a set of statistics that is sufficient for parameter estimation across models, this does not guarantee that the same set of statistics are sufficient for model selection. R \cite{Robert:2011wz} argue that therefore model selection, though not parameter estimation is fraught with problems in an ABC framework. Here we will argue that this is not the case. While we do agree that sufficiency and problems when using inadequate (or {\em insufficient}) statistics for model selection, we maintain that
\begin{itemize}
\item this mirrors problems that can also be observed in the parameter estimation context (see Figure \ref{fig:stat-param}),
\item for many important, and arguably the most important applications of ABC, this problem can in principle be avoided by using the whole data rather than summary statistics,
\item in cases where summary statistics are required, we argue that we can construct approximately sufficient statistics in a disciplined manner,
\item when all else fails, a change in perspective, allows us to nevertheless make use of the flexibility of the ABC framework (see Discussion).
\end{itemize}
The bulk of this article will deal with the derivation and application of a method that constructs summary statistics appropriate for the twin use of parameter estimation and model selection. After this we will briefly return to an outline of alternative approaches and map the applicability of ABC-based model selection more generally.

\section{Some Basic Concepts from Information Theory}
Our method is most easily framed in the terms of information theory \citep{Mackay:2003aa,Cover:2006aa,Mezard:2009aa}, and in order to keep this article self-contained we briefly review some of the basic concepts. We let $X$ denote a discrete random variable with potential states $\mathcal{X}$ and probability mass function $p_X(x)=\pr\left(X=x\right)$, $x\in\mathcal{X}$. In this section, for sake of clarity, we explicitly denote the random variables as subscript of the corresponding probability density.

\subsection{Entropy and Mutual Information}
The {\em entropy} of $X$, denoted by $H$, measures the uncertainty of $X$ and is defined as follows,
$$
H(X) = - \sum_{x} p_X(x) \log p_X(x) = - E_{X}\left[ \log p_X(X) \right] = E_{X} \left[ \frac{1}{\log p_X(X)} \right]\geq 0,
$$
where $E_{X}$ denotes the expectation under the probability mass function $p_X$. 
Let $(X,Y)$ be a pair of discrete random variables with joint distribution $p_{X,Y}$. The {\em conditional entropy} $H(Y|X)$ is defined as
$$
H(Y|X)=-E_{X,Y}\left[ \log p_{Y|X}(Y|X) \right].
$$
The {\em mutual information} $I(X;Y)$ between two discrete random variables $X$ and $Y$ measures the amount of information that $Y$ contains about $X$. It can be seen as the reduction of the uncertainty about $X$ due to the knowledge of $Y$,
$$
I(X;Y) = H(X)-H(X|Y) = \sum_{x,y\in\mathcal{X}} p_{X,Y}(x,y) \log \frac{ p_{X,Y}(x,y) }{ p_X(x) p_Y(y) }= KL( p_{X,Y} || p_Xp_Y )\geq 0,
$$
where $KL(P||Q)$ refers to the {\em Kullback-Leibler} (KL) divergence between probabilities $P$ and $Q$. The mutual information $I(X;Y)$ is equal to $0$ if and only if the random variables $X$ and $Y$ are independent.
\par
The {\em conditional mutual information} of discrete random variables $X$, $Y$ and $Z$ is defined as
$$
I(X;Y|Z)=H(X|Z)-H(X|Y,Z);
$$
it is the reduction in uncertainty of $X$ due to knowledge of $Y$ when $Z$ is given. This quantity is zero if and only if $X$ and $Y$ are conditionally independent given $Z$, which means that $Z$ contains all the information about $X$ in $Y$. The conditional mutual information satisfies the  chain rule: for discrete random variables $X_1,X_2,\cdots,X_n$ and $Y$ we have
$$
I(X_1,X_2,\cdots,X_n;Y)=\sum_{i=1}^nI(X_i;Y|X_1,X_2,\cdots,X_{i-2},X_{i-1}).
$$
In the following we only consider entropy, rather than the differential entropy, which applies to continuous random variables.

\subsection{Data Processing Inequality and Sufficient Statistics}
The {\em data processing inequality} (DPE) states that for random variables $X$, $Y$, and $Z$ such that $X \rightarrow Y \rightarrow Z$, (i.e. $Y$ depends, deterministically or randomly, on $X$ and $Z$ depends on $Y$)
$$
 I(X;Y) \ge I(X;Z),
 $$
with equality only if $X \rightarrow Y \rightarrow Z$ forms a Markov Chain, which means that the random variables $X$ and $Z$ are conditionally independent given $Y$: $p_{X,Z|Y} = p_{X|Y}p_{Z|Y}$.
\par
Now consider a family of distributions $\{ f(\cdot|\theta)\}_{\theta\in\Theta}$ and let $X$ be a sample from a distribution in this family. Let $\St$ be a deterministic statistic and denote by $S$ the random variable such that $S=\St(X)$. Therefore $\theta \rightarrow X \rightarrow S$. By the DPE
$$
I(\theta; S) \le I(\theta;X).
$$
Analogously to the discussion above, a statistic $\St$ is said \textit{sufficient for parameter $\theta$} if and only if $S$ contains all the information in $X$ about $\theta$, \ie 
$$
I(\theta; S) = I(\theta;X)\quad\text{where $S=\St(X)$}
$$ 
Equivalently we may write: 
\begin{theorem}\label{result1}
$\St$ is a sufficient statistic for parameter $\theta$ if and only if
$$ I(\theta; X|S) = 0.$$
In that case, $$E_{\theta,X}\left[ \log\frac{p(\theta|X)}{p(\theta|S)}\right]=0. $$
\end{theorem}
From now on, we resume the use of typically bayesian notation (see section 2) where the random variables are no longer given as subscripts but are unambiguously inferred from context. For example, the density of probability $p(\theta|S)$ involved in the theorem designates both the posterior probability of the parameter $\theta$ conditional on $S$ and its value when the parameter is actually equal to $\theta$.
\begin{proof}
By definition, $\St$ is a sufficient statistic if and only if $I(\theta;X) =I(\theta;S)$.
$\St$ being a deterministic fonction of $X$, the previous equation is equivalent to
\begin{align*}
&I(\theta;X,S) - I(\theta;S) = 0\\ 
\Leftrightarrow \quad & H(\theta) -H(\theta|X,S) - H(\theta) + H(\theta|S) = 0 \\
\Leftrightarrow \quad& H(\theta|S) - H(\theta|X,S) = 0\\
\Leftrightarrow \quad&  I(\theta; X|S) = 0
\end{align*}
If $\St$ is a sufficient statistic for parameter $\theta$ then
\begin{align*}
&I(\theta;X)=I(\theta;S)\\
\Leftrightarrow \quad &\sum_{\theta,x}p(\theta,x)\log\frac{p(\theta,x)}{p(\theta)p(x)}=\sum_{\theta,s}p(\theta,s)\log\frac{p(\theta,s)}{p(\theta)p(s)}\\
\Leftrightarrow \quad &\sum_{\theta,x}p(\theta,x)\log\frac{p(\theta,x)}{p(\theta)p(x)}=\sum_{\theta,x}p(\theta,x)\log\frac{p(\theta,\St(x))}{p(\theta)p(\St(x))}\\
\Leftrightarrow \quad &\sum_{\theta,x}p(\theta,x)\log\frac{p(\theta,x)p(\theta)p(\St(x))}{p(\theta)p(x)p(\theta,\St(x))}=0\\
\Leftrightarrow \quad &\sum_{\theta,x}p(\theta,x)\log\frac{p(\theta|x)}{p(\theta|\St(x))}=0
\end{align*}
\end{proof}

\section{Constructing Sufficient Statistics}
We consider the following situation: suppose that we have a finite set of summary statistics $\St=\left\{\St_1,\dots,\St_w\right\}$ and assume that $\St$ is a sufficient statistic. We aim to identify a subset $\U$ of $\St$ which is sufficient for $\theta$. The following result characterizes such a subset.  

\begin{theorem}
Let $\St$ be a finite set of summary statistics of $X$, and assume that $\St$ is a sufficient statistic. Denote by $\U$ a subset of $\St$. For a random variable $X$ distributed according to a distribution parametrized by $\theta$, let $U=\U(X)$ and $S=\St(X)$. The following statements hold
\begin{align*}
 &\text{$\U$ is a sufficient statistic}\\
\Leftrightarrow \quad & I(\theta;S|U)=0\\
\Leftrightarrow \quad & E_{X}\left[KL(p(\theta|S)||p(\theta|U))\right]=0\;.\\ 
\end{align*}
\end{theorem}
 \begin{proof}
By definition of the conditional mutual information 
\begin{align*}
I(\theta;S|U)&=H(\theta|U)-H(\theta|S,U)=H(\theta)-H(\theta|S,U)-H(\theta)+H(\theta|U)\\
&=I(\theta;U,S)-I(\theta;U)=I(\theta;S)-I(\theta;U)
\end{align*}
since $U$ is a vector composed of elements of $S$. The statistic $\St$ being sufficient, $I(\theta;S)=I(\theta;X)$, and therefore $I(\theta;S|U)=0$ if and only if $\U$ is a sufficient statistic.
Denote by $h$ the function such that for all $x$, $u=\U(x)=h(S(x))$,
\begin{align*}
 I(\theta;S|U) &=H(\theta|U)-H(\theta|S,U)=\sum_{\theta,s,u}p(\theta,s,u)\log\left(\frac{p(\theta|u,s)}{p(\theta|u)}\right)\\
 &=\sum_{\theta,s}p(\theta,s)\log\left(\frac{p(\theta|h(s),s)}{p(\theta|h(s))}\right)=\sum_{\theta,s}p(\theta,s)\log\left(\frac{p(\theta|s)}{p(\theta|h(s))}\right)\\
 &=\sum_{s}p(s)\sum_\theta p(\theta|s)\log\left(\frac{p(\theta|s)}{p(\theta|h(s))}\right)=\sum_{s}p(s)KL\left(p(\theta|s)|| p(\theta|h(s)) \right)\\
 &=\sum_{x}p(x)KL\left(p(\theta|\St(x))|| p(\theta|\U(x)) \right)=E_{X }\left[KL\left(p(\theta|S)|| p(\theta|U) \right)\right]\;.
\end{align*}
\end{proof}

According to this result identifying a sufficient statistic with minimum cardinality from a sufficient family $\St$ of summary statistics boils down to identifying the smallest subset $\U$ of $\St$ such that $I(\theta;S|U)=0$ where $U=\U(X)$ and $S=\St(X)$ or, equivalently, 
$E_{X }\left[KL(p(\theta|S)||p(\theta|U))\right]=0$. 
Here we aim to determine a sufficient statistic for Approximate Bayesian Computation (ABC) methods for parameter inference and then for model selection. We focus in this section on the parameter inference task.  
In the ABC framework, the expectation over the data of the Kullback-Leibler divergence \citep{Cover:2006aa} between the two posterior distributions $p(\theta|S)$ and $p(\theta|U)$ cannot be exactly computed since we only have at our disposal a dataset $x$ and the value of the statistics for this dataset $s^*=(s_1^*,\dots,s_w^*)=\St(x)$. Thus, we approximate it by the expectation with respect to the empirical measure of the data. The method is summarized in Algorithm~\ref{algo:min}. We denote by $|U|$ the cardinality of a set $U$. 
\begin{algorithm}[h!]
  \begin{algorithmic}[1]
    \caption{Minimization of the mutual information}
    \label{algo:min}
    \State {\bfseries input:} a sufficient set of statistics whose values on the dataset is $s^*=\left\{s_1^*,\dots,s_w^*\right\}$
    \State {\bfseries output:} a subset $U^*$ of $s^*$
    \For{ all $u^*\subset s^*$} 
    \State perform ABC to obtain $\hat{p}(\theta|u^*)$
    \EndFor
    \State let $T^* = \left\{ u^*\subset s^*\text{ such that } KL\left(\hat{p}(\theta|s^*)|| \hat{p}(\theta|u^*)\right)=0  \right\}$
    \State  {\bfseries return } $U^*=\argmin_{u^*\in T^*} |u^*|$
  \end{algorithmic}
\end{algorithm}

This methodology is computationally prohibitive since it enumerates all possible subsets $u^*$ of $s^*$ and perform the ABC algorithm for all of these. 
Moreover, it is challenging to obtain a precise estimate of the posterior $\hat{p}(\theta|u^*)$ of $\theta$ given a value of the statistics $U=u^*$ when the cardinality of $u^*$ is large. In particular, it is often impossible to obtain an estimate of $p(\theta|s^*)$. It is then necessary to design an algorithm which does not need the computation of this probability. The following result provides a first step into this direction.

\begin{theorem}\label{result:dim}
Let $X$ be a random variable generated according to $f(\cdot|\theta)$. Let $\St$ be a sufficient statistic and $\U$ and $\T$ two subsets of $\St$ such that $U=\U(X)$, $T=\T(X)$ and $S=\St(X)$ satisfy $U\subset T \subset S$. We have 
$$I(\theta;S|T)=I(\theta;S|U)-I(\theta;T|U) \;.$$
\end{theorem}
\begin{proof}
For all subset $T$ of $S$,
$$ I(\theta;S|T)=H(\theta|T)-H(\theta|S,T)=H(\theta|T)-H(\theta|S)\;.$$
Then, for $U$ and $T$ subsets of $S$,
$$ I(\theta;S|T)-I(\theta;S|U)=H(\theta|T)-H(\theta|U)\;.$$
If $U$ in included in $T$, then $H(\theta|T)=H(\theta|T,U)$ and
$$ I(\theta;S|T)-I(\theta;S|U)=-I(\theta;T|U)$$
which proves the result.
\end{proof}
It follows that the information contained in $S$ on $\theta$ given $T$ is smaller than the information contained in $S$ on $\theta$ given $U\subset T$. In order to construct a subset $\T$ of $\St$ such that $I(\theta;S|T)=0$, where $S=\St(X)$ and $T=\T(X)$, it is thus sufficient to add one by one statistics of $\St$ until the condition holds. Indeed, if we denote the statistic added at time step $k\le w$ by $\St_{(k)}$ and $S_{(k)}=\St_{(k)}(X)$, then
$$
I(\theta;S|S_{(1)},\dots,S_{(k)})\leq I(\theta;S|S_{(1)},\dots,S_{(k+1)})\;.
$$
And there exists an integer $k\leq w$ such that $I(\theta;S|S_{(1)},\dots,S_{(k)})=0$. According to result~\ref{result:dim}, at each time $k$, 
$$
I(\theta;S|S_{(1)},\dots,S_{(k)})=I(\theta;S|S_{(1)},\dots,S_{(k-1)})-I(\theta;S_{(1)},\dots,S_{(k)}|S_{(1)},\dots,S_{(k-1)}) \;.
$$ 
The mutual information is a non negative function, then, in order to decrease as much as possible $I(\theta;S|S_{(1)},\dots,S_{(k)})$, the added statistic at time $k> 1$ should be such that 
\begin{align*}
S_{(k)}&=\argmax_{V\in S\backslash \{S_{(1)},\dots,S_{(k-1)}\}}I(\theta;S_{(1)},\dots,S_{(k-1)},V|S_{(1)},\dots,S_{(k-1)})\\
&= \argmax_{V\in S\backslash \{S_{(1)},\dots,S_{(k-1)}\}}E_{X}\left[KL\left(p(\theta|S_{(1)},\dots,S_{(k-1)},V)||p(\theta|S_{(1)},\dots,S_{(k-1)})  \right)\right]\;.
\end{align*}
As previously mentioned, in practice, this expectation can not be computed. We hence replace it by the expectation with respect to the empirical measure of the data, leading to the approximation, for all $1\leq k\leq n$,
\begin{multline*}
 \argmax_{V\in S\backslash \{S_{(1)},\dots,S_{(k-1)}\}}E_{p(X)}\left[KL\left(p(\theta|S_{(1)},\dots,S_{(k-1)},V)||p(\theta|S_{(1)},\dots,S_{(k-1)})  \right)\right]\\ \approx \argmax_{v^*\in s^*\backslash \{s_{(1)}^*,\dots,s_{(k-1)}^*\}}KL\left(p(\theta|s_{(1)}^*,\dots,s_{(k-1)}^*,v^*)||p(\theta|s_{(1)}^*,\dots,s_{(k-1)}^*)  \right)\;,
\end{multline*}
where $s^*$ is the value of the statistic $\St$ on the dataset. In addition, the first statistic $\St_{(1)}$ should contain the maximum information about $\theta$, in the sense that
$$
S_{(1)}=\argmax_{V\in S} I(\theta;V) = \argmin_{V\in S} H(\theta|V)=\argmax_{V\in S}E_{p(\theta,V)}\left[ \log p(\theta|V) \right]\;. 
$$ 
This results in algorithm~\ref{algo1}.

\begin{algorithm}[h!]
  \begin{algorithmic}[1]
    \caption{Greedy minimization of the mutual information}
    \label{algo1}
    \State {\bfseries input:} a sufficient set of deterministic statistics whose values on the dataset is $s^*=\left\{s_1^*,\dots,s_w^*\right\}$
\State {\bfseries output:} a subset $U^*$ of $s^*$
    \State for all $u^*\in s^*$, perform ABC to obtain $\hat{p}(\theta|u^*)$
    \State let $s^*_{(1)} = \argmax_{u^*\in s^*} \log \hat{p}(\theta|u^*)$
    \For{$k\in \{2,\dots,w\}$}
    \State for all $u^*\in s^* \backslash \left\{s_{(1)}^*,\dots,s_{(k-1)}^*\right\}$, perform ABC to obtain $\hat{p}(\theta|s_{(1)}^*,\dots,s_{(k-1)}^*,u^*)$
    \State let 
\begin{equation}
s_{(k)}^*=\argmax_{u^*}KL(\hat{p}(\theta|s_{(1)}^*,\dots,s_{(k-1)}^*,u^*)||\hat{p}(\theta|s_{(1)}^*,\dots,s_{(k-1)}^*))
\label{eq:maxstep}
\end{equation}
\If{$KL(\hat{p}(\theta|s_{(1)}^*,\dots,s_{(k)}^*)||\hat{p}(\theta|s_{(1)}^*,\dots,s_{(k-1)}^*))\leq\epsilon $}
\State {\bfseries return} $U^*=(s_{(1)}^*,\dots,s_{(k-1)}^*))$
\EndIf
\EndFor
\State {\bfseries return} $U^*=s^*$

  \end{algorithmic}
\end{algorithm}

This algorithm is computationally expensive if the cardinality of the set of statistics $|\St|=w$ is large. Indeed at each iteration $k$, one performs the ABC algorithm $w-k+1$ times. A simplification of it consists in replacing the maximization step (see equation~\eqref{eq:maxstep}) by testing randomly various statistics and choosing a statistic $\St_{(k)}$ such that $I(\theta;S_{(1)},\dots,S_{(k)}|S_{(1)},\dots,S_{(k-1)})$ is large. Different criteria may be used to determine if the mutual information is large or not, and then decide if the statistic should be included or not. Most of these criteria consist of determining if the posterior probability of $\theta$ given $S_{(1)},\dots,S_{(k-1)}$ and the posterior probability of $\theta$ given $S_{(1)},\dots,S_{(k)}$ are significantly different. If so, adding the statistic $\St_{(k)}$ is justified, otherwise we do not add it and instead turn to a different statistic. In the algorithm~\ref{algo2}, we denote by $\mathcal{C}\left[p(\theta|s_{(1)}^*,\dots,s_{(k)}^*),p(\theta|s_{(1)}^*,\dots,s_{(k-1)}^*)\right]$ the criterion which is equal to $1$ if the statistic $\St_{(k)}$ should be added and $0$ otherwise. 
In practice, we can use, for instance, one of the following measures:
\begin{itemize}
 \item we add the statistic $\St_{(k)}$ if $KL(p(\theta|s_{(1)}^*,\dots,s_{(k)}^*)||p(\theta|s_{(1)}^*,\dots,s_{(k-1)}^*))\geq \delta_k$ where $\delta_k$ is a threshold which may be computed by bootstrapping the data to estimate  $E_{p(X)}\left[p(\theta|S_{(1)},\dots,S_{(k-1)})\right]$: 
 $$
 \delta_k= KL\left(p(\theta|s_{(1)}^*,\dots,s_{(k)}^*)||E_{p(X)}\left[p(\theta|S_{(1)},\dots,S_{(k-1)})\right]\right)\;,
 $$
\item the Kolmogorov-Smirnov test enables us to compare $p(\theta|s_{(1)}^*,\dots,s_{(k-1)}^*)$ and $p(\theta|s_{(1)}^*,\dots,s_{(k)}^*)$ and so, the statistic $\St_{(k)}$ is added if the test has a $p$-value smaller than a certain threshold, say $0.01$.
\end{itemize}
\begin{algorithm}[h!]
  \begin{algorithmic}[1]
    \caption{Stochastic minimization of the mutual information}
    \label{algo2}
    \State {\bfseries input:} a sufficient set of deterministic statistics whose values on the dataset is $s^*=\left\{s_1^*,\dots,s_w^*\right\}$
\State {\bfseries output:} a subset $V^*$ of $s^*$
    \State choose randomly $u^*$ in $s^*$
\State $T^*\leftarrow s^* \backslash \left\{u^*\right\}$
\State $V^*\leftarrow u^*$
    \Repeat
    \Repeat
    \If{$T^*=\O$} {\bfseries return} $V^*$\EndIf
    \State choose randomly $u^*$ in $T^*$
    \State $T^*\leftarrow T^*\backslash u^*$
    \State perform ABC to obtain $\hat{p}(\theta|V^*,u^*)$
\Until{ $\mathcal{C}\left[\hat{p}(\theta|V^*,u^*),\hat{p}(\theta|V^*)\right]=1$}
\State $T^*\leftarrow s^* \backslash \left\{V^*,u^*\right\}$
\State {\itshape optionally:} $V^*\leftarrow\textit{OrderDependency}\left(V^*,u^*\right)$ and $T^*\leftarrow s^* \backslash V^*$
\State $V^*\leftarrow V^*\cup u^*$
\Until {$T^*=\O$} 
\State{\bfseries return} $V^*$
  \end{algorithmic}
\end{algorithm}
The order in which the statistics are added matters \citep{Joyce:2008vm,Nunes:2010dv} and the only way to avoid this inconvenience is to use the computationally expensive algorithm~\ref{algo1}. Nevertheless, we suggest to test for order dependency before deciding on whether to add a statistic. This consists of determining if, given the recently added statistic $u^*$, the previously added statistics in $V^*$ still bring relevant information or are not necessary anymore and hence may be released from the set under construction. More precisely, after line 15 of the algorithm~\ref{algo2}, we add the function described in algorithm~\ref{algo:BackRej}.
\begin{algorithm}[h!]
  \begin{algorithmic}[1]
    \caption{Order Dependency}
    \label{algo:BackRej}
    \State {\bfseries Input:} A set of accepted statistics $V^*=\{s^*_{(1)},\dots,s^*_{(k-1)}\}$ and the last accepted statistic $u^*$
    \State {\bfseries Output:} A subset $U^*$ of $\left\{V^*\cup u^*\right\}$
    \State $U^*\leftarrow u^*$
    \For{$i\in\{1,\dots,k-1\}$}
    \If{$\mathcal{C}(p(\theta|U^*,s_{(i)}^*),p(\theta|U^*))=1$}
    \State $U^*\leftarrow U^*\cup s_{(i)}^*$
    \EndIf
    \EndFor
\State{\bfseries return} $U^*$
  \end{algorithmic}
\end{algorithm}

\section{Relation to previous work}
The presented information theoretic framework builds upon two previous methods for summary statistic selection. Our main contributions are the generalisation of the notion of approximate sufficiency, rigorous derivations of algorithms using information theory and the application of summary statistic selection for the joint space.
\par
\cite{Joyce:2008vm} developed a notion of approximate sufficiency for parameter inference and presented a sequential algorithm to score statistics according to whether their inclusion would improve the inference. Their sequential algorithm resembles (and indeed inspired) Algorithm 3 although they do not retain statistics once they have failed to be added which we feel is required since whether a statistic is added depends strongly on the statistics already accepted. Their rule for adding statistics is essentially an approximate test for independence on the posterior distribution under the addition of a new statistic but can only be used for single parameter models. We have shown that the true stopping criterion should be the change in KL divergence which can be used for multivariate posteriors, although tests for independence (KS, $\chi^2$) can be used as an approximation in single parameter models.
\par
\cite{Nunes:2010dv} proposed a heuristic algorithm to minimize the entropy of the posterior with respect to sets of summary statistics \cite{}. Additionally they proposed a refinement step where the set of statistics that minimised the posterior mean root sum of squared errors (MRSSE) was selected. The minimum entropy approach is related to Algorithm 1 since, when $H(\theta)$ is constant, minimising the entropy maximises the mutual information. However, assuming there exists a sufficient statistic, choosing the set of statistics that minimises the entropy is guaranteed to give sufficiency but not {\em minimal} sufficiency since adding a statistic to a sufficient set can reduce the entropy by chance (a manifestation of ``conditioning always reduces entropy").

\section{Automated selection of summary statistics for model selection}
As pointed out recently sufficiency across models is still not sufficient to perform reliably model choice in the ABC framework. Here we show how a natural extension of the methodology introduced above can also be employed in order to construct sets of statistics that are sufficient for ABC-based model selection, when it is impractical to use the raw data \citep{Toni:2010kt}.
\par
Consider $q$ models, each with an associated set of parameters $\Theta_i, i \in \{1, ..., q\}$. We aim to identify a set of sufficient statistic for model selection. Let $M$ being a random variable taking value in $\{1,\dots,q\}$. A statistic is sufficient for model selection if and only if it is sufficient for the joint space $\{M,\{\theta_i\}_{1\leq i\leq q}\}$. According to result~\ref{result1}, this means that a summary statistic $\St$ is sufficient for model selection if and only if $I(M, \theta_{1}, ...., \theta_{q};X|S)=0$ where $S=\St(X)$ and $X$ is a sample from a distribution in the family $\{f(\cdot|\theta_i)\}_{\theta_i\in\Theta_i,\: 1\leq i\leq q}$. The following result enables us to link this condition with the sufficiency for parameter inference for each model.
\begin{theorem}
For all deterministic statistic $\St$, 
$$ I(M, \theta_{1}, ...., \theta_{q};X|S) = I(M; X | \theta_{1}, ...., \theta_{q},S) + \sum_{i} I( \theta_{i}; X|S)\;, $$
where $S=\St(X)$ and $X$ is a sample from a distribution in the family $\{f(\cdot|\theta_i)\}_{\theta_i\in\Theta_i,\: 1\leq i\leq q}$.
\end{theorem}
\begin{proof}
From the chain rule of mutual information we have that
$$I(M, \theta_{1}, ...., \theta_{q};X|S) = I(M; X | \theta_{1}, ...., \theta_{q},S)+\sum_{i} I( \theta_{i}; X|\theta_1,\cdots,\theta_{i-1},S)\;.$$
The result then follows from the fact that, for all $i$, $I(\theta_{i}; X | \theta_1,\cdots,\theta_{i-1},S) = I(\theta_{i}; X|S)$.
\end{proof}
The mutual information being always non negative, this shows that a statistic $\St$ is sufficient for model selection if and only if $I(M; X | \theta_{1}, ...., \theta_{q},S)=0$ and $I( \theta_{i}; X|S)=0$ for all $i\in\{1,\cdots,q\}$. Therefore, a sufficient statistic for model selection is sufficient for parameter inference in each model and, given the parameter values for every models, the statistic is sufficient for inferring the model. Thus, in order to identify a sufficient statistic for model selection, one should determine a set of minimal sufficient statistics $\St^{m_i}$ for each model $1\leq i\leq q$ and then identify among all statistics $\T$ containing $\cup_{i\in\{1\cdots q\}}\St^{m_i}$ the one with the smallest cardinality such that $I(M; X | \theta_{1}, ...., \theta_{q},T)=0$. The method, summarized in algorithm~\ref{algo:model}, consists in running one of the previous algorithms, for example algorithm~\ref{algo2} and then add statistics which bring new information about the models in the sense that the posterior probability of the model conditional on the  statistics varies significantly if we add a this new statistic.

\begin{algorithm}[h!]
  \begin{algorithmic}[1]
    \caption{Stochastic minimization of the mutual information for model selection}
    \label{algo:model}
    \State {\bfseries input:} a sufficient set of deterministic statistics whose values on the dataset is $s^*=\left\{s_1^*,\dots,s_w^*\right\}$
 \State {\bfseries output:} a subset $V^*$ of $s^*$ which is sufficient for model selection
\For{$i\in\{1,\dots,q\}$}
\State determine a sufficient statistic whose value on the dataset is $S^{m_i}\subset s^*$ using algorithm~\ref{algo2}
\EndFor
\State Let $M^*=\cup_{1\leq i\leq q}S^{m_i}$
\State Let $W^*\leftarrow s^*\backslash M^*$
    \State choose randomly $u^*$ in $W^*$
\State $T^*\leftarrow W^* \backslash \left\{u^*\right\}$
\State $V^*\leftarrow u^*$
    \Repeat
    
    \Repeat
    \If{$T^*=\O$} {\bfseries return} $V^*$
\EndIf
    \State choose randomly $u^*$ in $T^*$
    \State $T^*\leftarrow T^*\backslash u^*$
    \State perform ABC to obtain $\hat{p}(M|\theta_1,\dots,\theta_q,M^*,V^*,u^*)$
\Until{ $\mathcal{C}\left[\hat{p}(M|\theta_1,\dots,\theta_q,M^*,V^*,u^*),\hat{p}(M|\theta_1,\dots,\theta_q,M^*,V^*)\right]=1$}
\State $T^*\leftarrow W^* \backslash \left\{V^*,u^*\right\}$
\State {\itshape optionally:} $V^*\leftarrow\textit{OrderDependenceModelSelec}\left(V^*,M^*\cup u^*\right)$ and $T^*\leftarrow W^* \backslash V^*$
\State $V^*\leftarrow V^*\cup u^*$
	\Until{$T^*=\O$} 
\State{\bfseries return} $V^*$

  \end{algorithmic}
\end{algorithm}
Similarly to algorithm~\ref{algo2}, it is possible to test for order dependency before deciding to add a statistic. To do so, we apply the algorithm~\ref{algo:BackRej} in which at step 3,  the set $U^*$ is initialized by $M^*\cup u^*$ such that we always keep the set $M^*$ and condition line 5 is replaced by $\mathcal{C}\left[\hat{p}(M|\theta_1,\dots,\theta_q,U^*,s_{(i)}^*),\hat{p}(M|\theta_1,\dots,\theta_q,U^*)\right]=1$.

\section{Applications}
We illustrate the framework developed above in three different contexts. First we consider a simple model selection problem involving two normal distributions. We then consider  a typical population genetics example on three demographic scenarios for simulated data, before finally turning to a problem where we consider alternative random walk models; this last example should be typical for applications where likelihood-based inferences are out of the question due to the complexity of the models and the data.
\subsection{Normal example}
The developed framework was illustrated on a simple example with two models
\begin{equation*}
y_1, ... y_d \sim \mathcal{N}(\mu,\sigma^2_{1}) \; \text{ and } \; y_1, ... y_d \sim \mathcal{N}(\mu,\sigma^2_{2}),
\end{equation*}
where the variances of the normal distributions, $\sigma_1$ and $\sigma_2$, are fixed. Under a conjugate prior, $\mu \sim \mathcal{N}(0,a^2)$, the true Bayes factor is given by
\begin{equation*}
BF( \mathbf{y} )= \frac{ \sigma_{1}^{-d} \exp\{ -S^2/2\sigma_{1}^{2}\}\exp\{ -\bar{y}^2/2(a^{2} + \sigma_{1}^{2}/d )\} \sqrt{ a^{-2} + d\sigma_{2}^{-2} } }{\sigma_{2}^{-d} \exp\{ -S^2/2\sigma_{2}^{2}\}\exp\{ -\bar{y}^2/2(a^{2} + \sigma_{2}^{2}/d )\} \sqrt{ a^{-2} + d\sigma_{1}^{-2} } },
\end{equation*}
where
\begin{equation*}
\bar{y} = d^{-1} \sum_{i=1}^d y_i \; \text{ and } \; S^2 = \sum_{i=1}^{d} (y_i - \bar{y})^2.
\end{equation*}
In this case $\bar{y}$ is sufficient for parameter inference but the pair $\{ \bar{y}, S^2 \}$ is sufficient for the joint space.
\par
To test the automated choice of approximate summary statistics, 100 data sets were sampled under model 1 and the algorithm run to select statistics for parameter inference and for the joint space from a pool of 5 statistics including $\bar{y}$, $S^2$, range, maximum and a non informative statistic $u \sim \mathcal{U}(0,2)$. The values of the parameters were chosen to be $\sigma_1 = 0.3$, $\sigma_2 = 0.6$, $n = 15$ and $a=2$. Since there are difficulties that arise from the possibly very different scales of the summary statistics the distance was defined as
$$ \Delta( \St(x), \St(y) ) = \sum_{i} [ \log \St_{i}(x) - \log \St_{i}(y) ]^2,$$
which accounts for the relative difference between data and simulation and avoids the need to known the scales of the statistics {\em a priori}. Here, $\St_i(x)$ denotes the $i$-th component of $\St(x)\in\mathds{R}^w$The stopping criteria were defined through tests for independence ($p<1\times10^{-5}$) between the posterior distributions under different summary statistics; a Kolmogorov-Smirnov test in the case of the continuous parameter posterior and a Pearson test in the case of the discrete model posterior. The ABC was run with 500 particles with $\epsilon = 0.1$.
\par
Figure \ref{fig:NormalStats} shows the results of the summary statistic selection over the 100 runs. Figure \ref{fig:NormalStats} (left) shows the results for parameter inference across the two models. The mean was selected in every replicate, the maximum in 14 replicates and $S^2$ once. Figure \ref{fig:NormalStats} (right) shows the additional statistics selected for the joint space. $S^2$ was selected in 84 cases, the range in 19 cases, the maximum in 9 cases and the noise statistic once. Figure \ref{fig:NormalBF} shows the Bayes Factor obtained via ABC versus the analytical prediction. As expected the Bayes factor calculated using only the statistics selected for parameter inference is uncorrelated with the true Bayes factor (figure on the left). When the statistics sufficient for the joint space are included the ABC Bayes factor correlates well with the analytical prediction.
\begin{figure}[p]
\begin{center}
\includegraphics[width=0.45\textwidth]{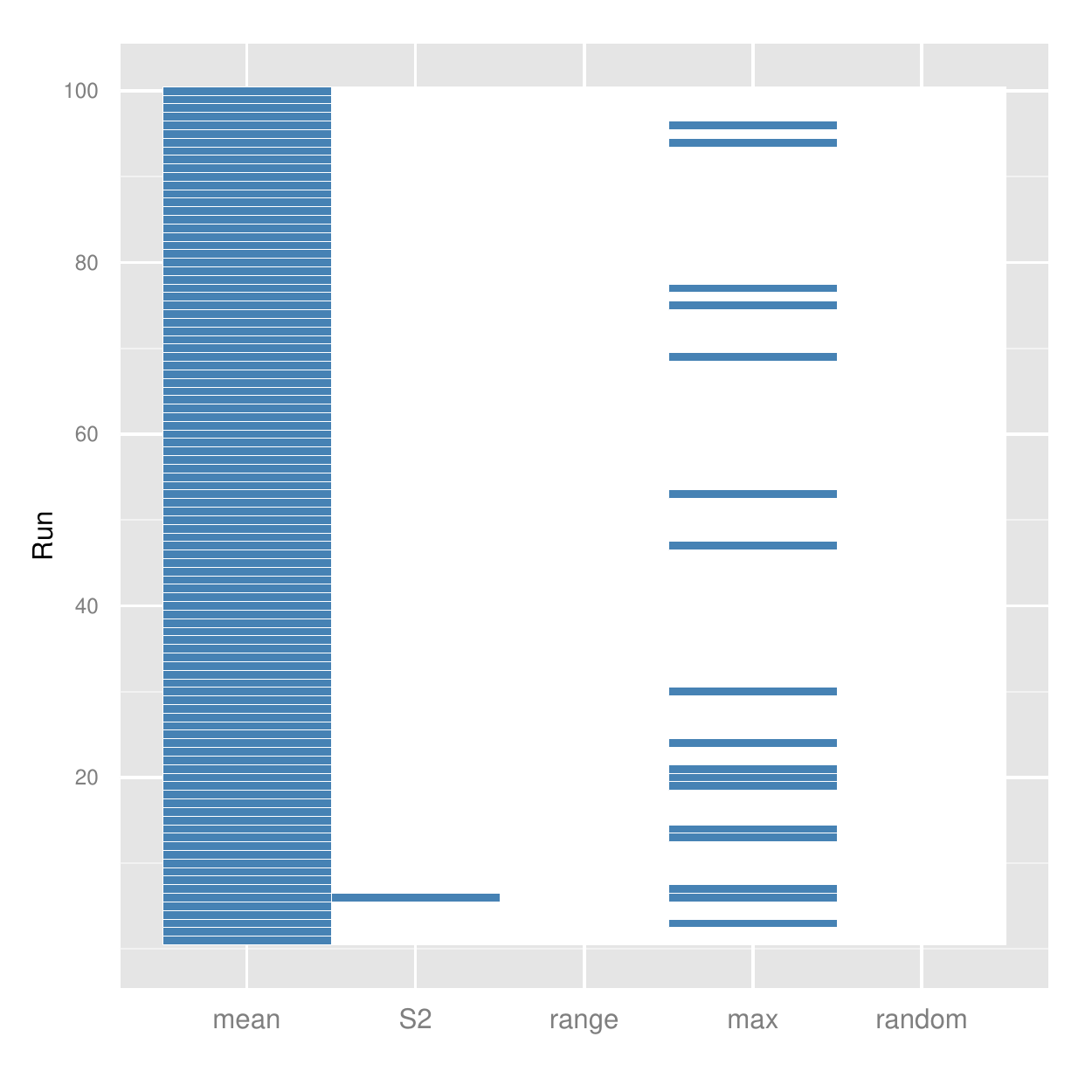}
\includegraphics[width=0.45\textwidth]{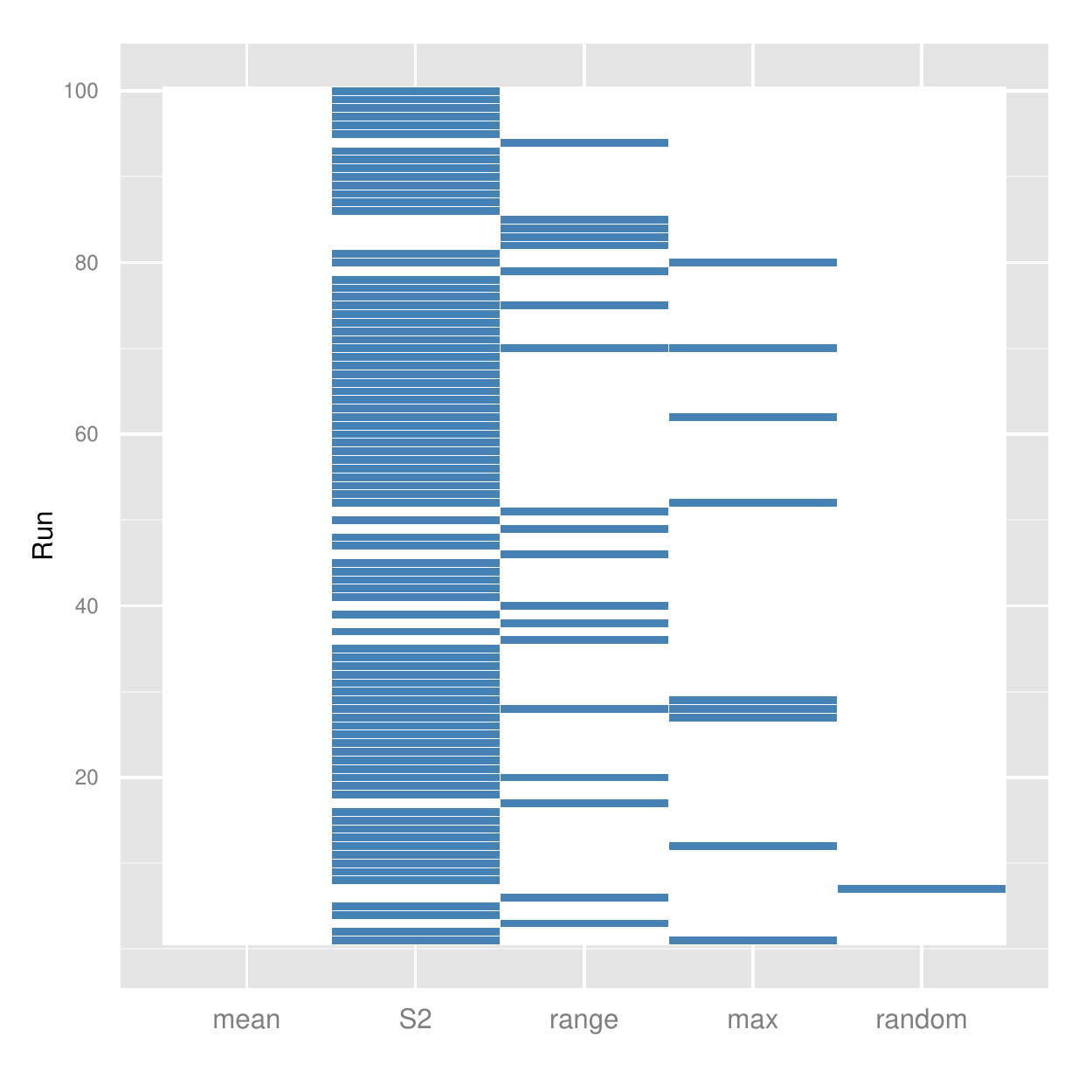}
\caption{ \label{fig:NormalStats} Summary statistics selected in 100 runs of the automated summary statistic selection. Left: Statistics selected for parameter inference (the union of statistics found under model 1 and model 2). Right: Additional statistics chosen for the joint space.  }
\end{center}
\end{figure}

\begin{figure}[p]
\begin{center}
\includegraphics[width=0.45\textwidth]{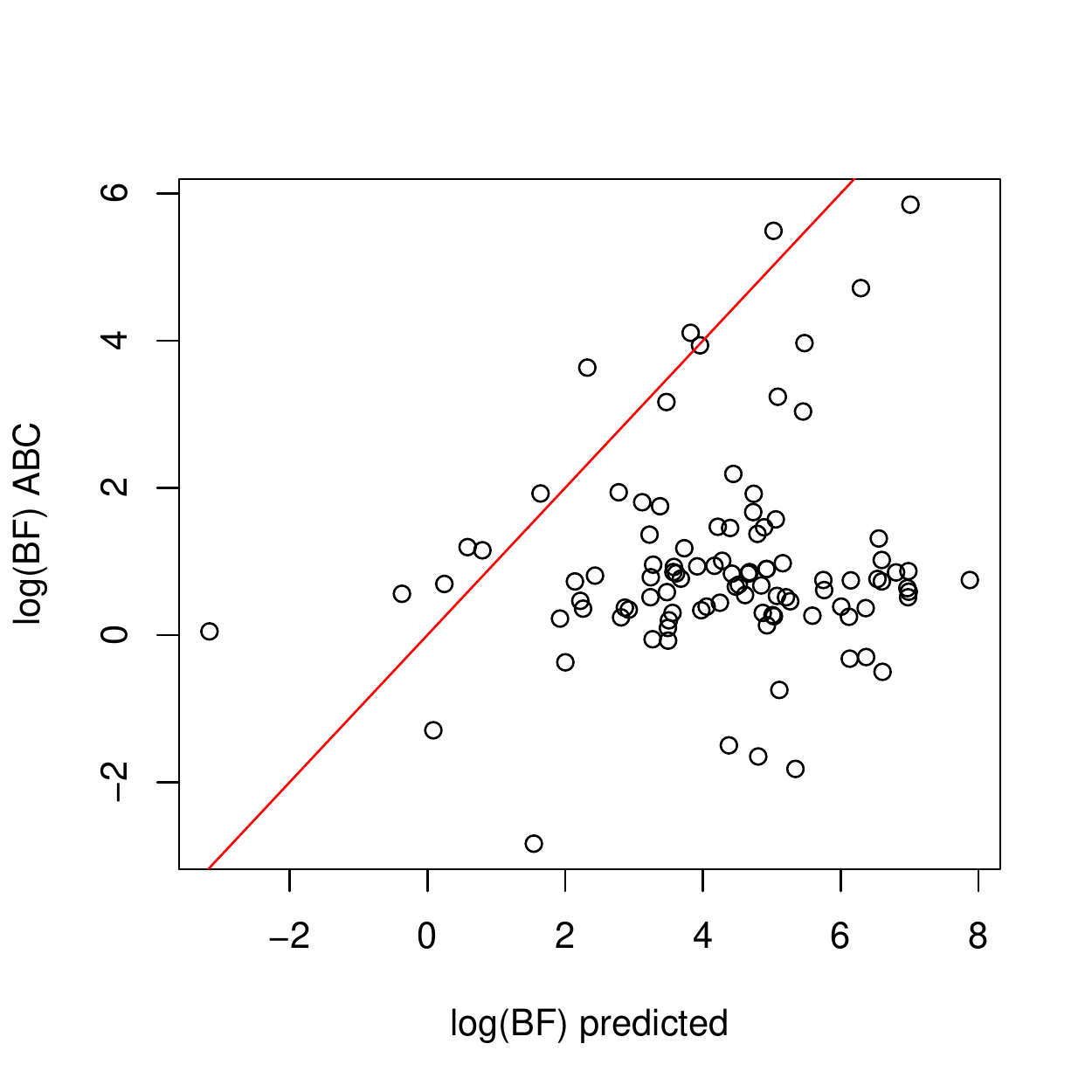}
\includegraphics[width=0.45\textwidth]{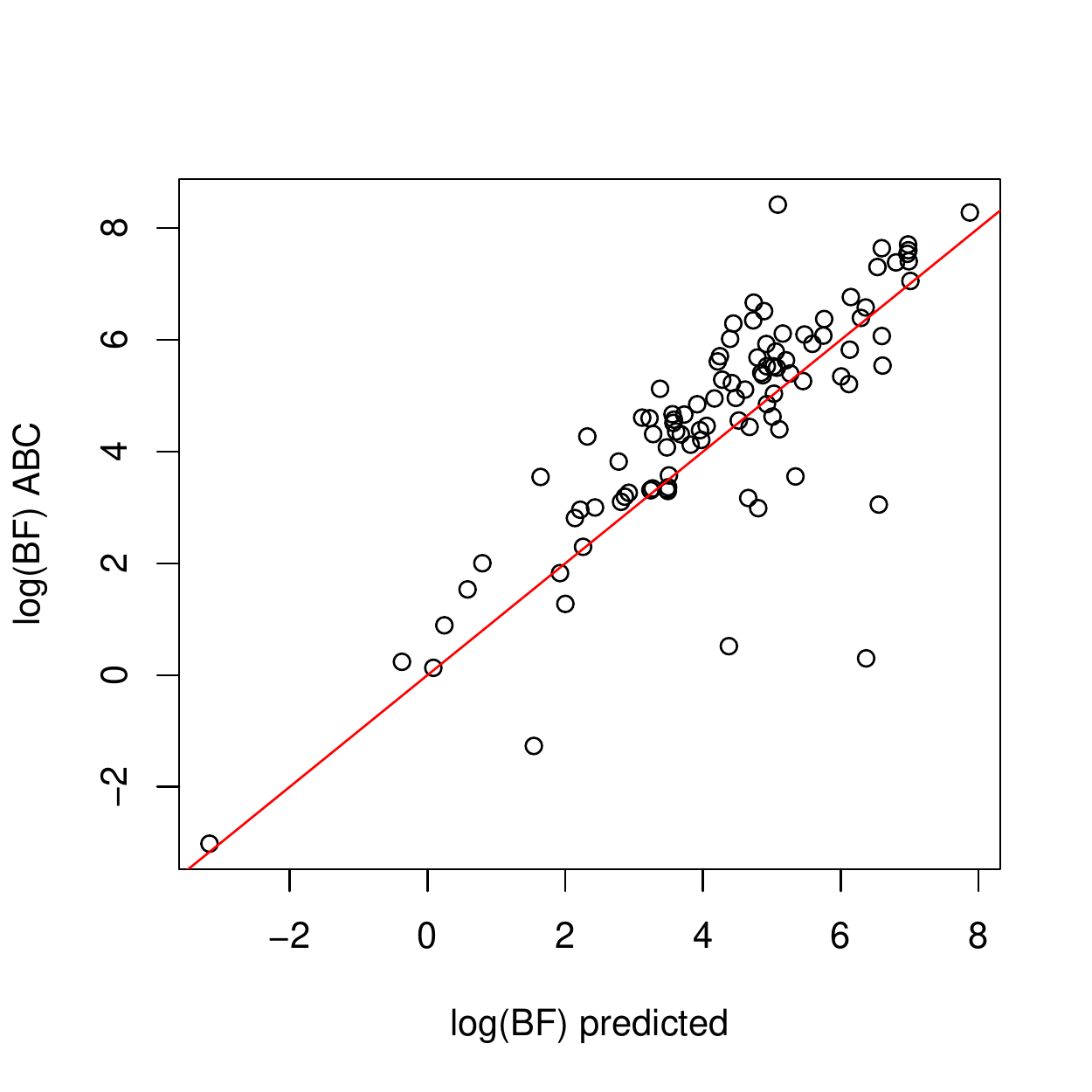}
\caption{ \label{fig:NormalBF} Predicted vs approximated log Bayes Factor for the normal toy model. Left: The case for sufficient statistics selected for parameter inference. Right: The case for sufficient statistics selected for the joint space. In both cases the red line represents the line $y=x$. }
\end{center}
\end{figure}

\subsection{Population genetics example}

To further demonstrate the efficacy of our methodology we applied the summary statistic selection procedure to a real world ABC problem, that of model selection in population genetics. Data were generated using coalescent simulations \citep{Hudson:1991vo} from three competing models, producing $100$ data sets from each for fixed parameters using a modified version of the MS software package downloadable from \url{http://home.uchicago.edu/rhudson1/source/mksamples.html}).

The models considered were:
\begin{description}
\item[Model 1] Constant population size (with population mutation rate $\theta=20$, corresponding to a 20,000bp stretch of DNA in a population of size $N=1,000,000$.
\item[Model 2] Exponential growth model with exponential growth rate $\gamma=0.4$ and all other parameters as above.
\item[Model 3] Two island model with scaled migration rate $m=10$ and all other parameters as above.
\end{description}

 To perform ABC we generated $5,000,000$ samples of datasets comprising 100 chromosomes for each model; in each case the population mutation rate $\theta$ \citep{Ewens2004} was drawn from the prior $\mathcal{U}(5,30)$; the real value for which the data were generated was $\theta=20$ (measured in units of total population size). The summary statistics summarised below were then used in our ABC summary statistic selection framework to derive sufficient sets of summary statistics for model selection on the observed data. The summary statistics calculated were:
\begin{description}
\item[S1] Number of Segregating Sites, $N_S$.
\item[S2] Number of Distinct Haplotypes, $N_H$.
\item[S3] Homozygosity, $h_H$, where $h$ is the probability that two haplotypes are identical,
$$
h_H= \sum_{h=1}^{N_H}\nu_h^2.
$$
\item[S4] Average SNP Homozygosity, $$
\bar{h}_S=\sum_{i=1}^{N_S} (\nu_0(i)^2+\nu_1(i)^2).
$$
\item[S5] Number of occurences of most common haplotype, $f_H$.
\item[S6] Mean number of pair-wise differences between haplotypes, $T$.
\item[S7] Number of Singleton Haplotypes, $f_{sH}$
\item[S8] Number of Singleton SNPs, $f_{sS}$.
\item[S9] Linkage Disequilibrium measured by 
$$
\overline{r^2}= \frac{2}{N_S(N_S-1)}\sum_{i=1}^{N_S-1}\sum_{j=i+1}^{N_S} \frac{\left(\nu_{00}(i,j)-\nu_0(i)\nu_0(j)\right)^2}{\nu_0(i)\nu_1(i)\nu_0(j)\nu_1(j)}
$$ 
\item[S10] Fraction of pairs of loci which violate the four-gamete test, i.e. for which the two-locus haplotypes $00$, $01$, $10$ and $11$ exist.
\item[S11] Random variable, $\rho\sim{\mathcal U}_{[0,1]}$.
\end{description}
where $N$ is the number of sequences in the data, and for each SNP locus, $i$, let $\nu_0(i)$ $\nu_1(i)$ denote the frequencies of the ancestral and derived alleles; further for any haplotype $h$, $\nu_h$ is the corresponding frequency.
\

\begin{figure}[p]
\begin{center}
\includegraphics[width=0.48\textwidth]{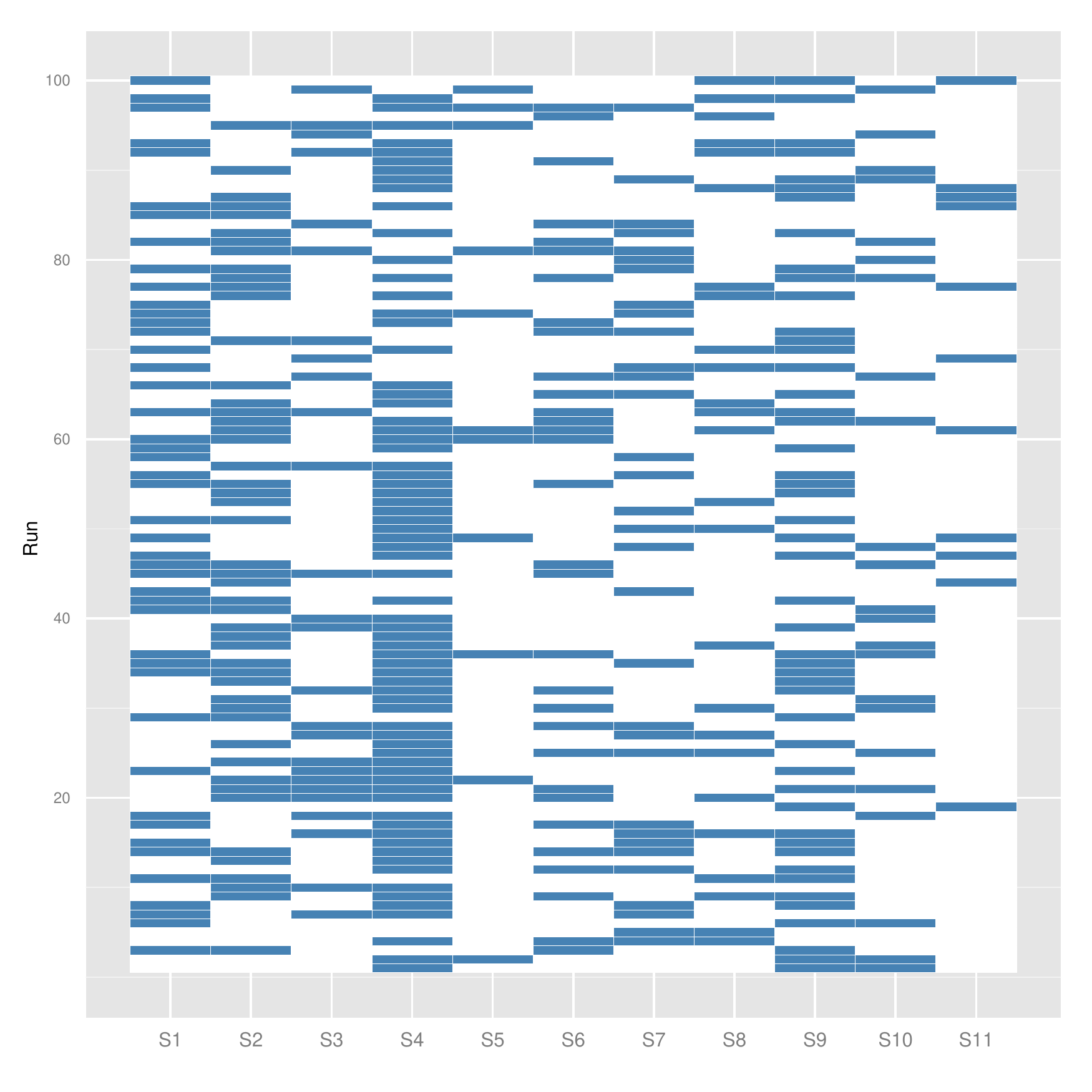}
\includegraphics[width=0.48\textwidth]{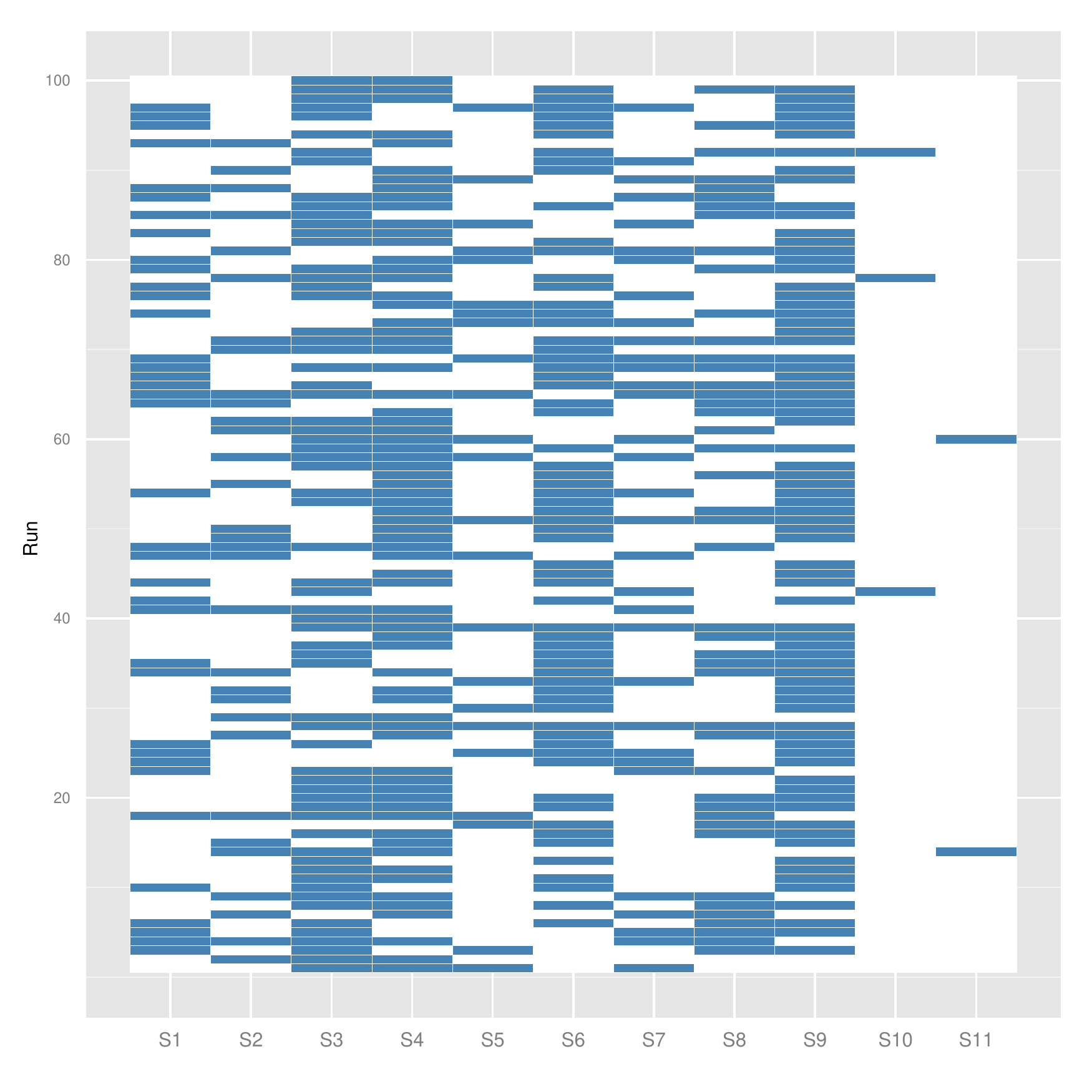}
\includegraphics[width=0.48\textwidth]{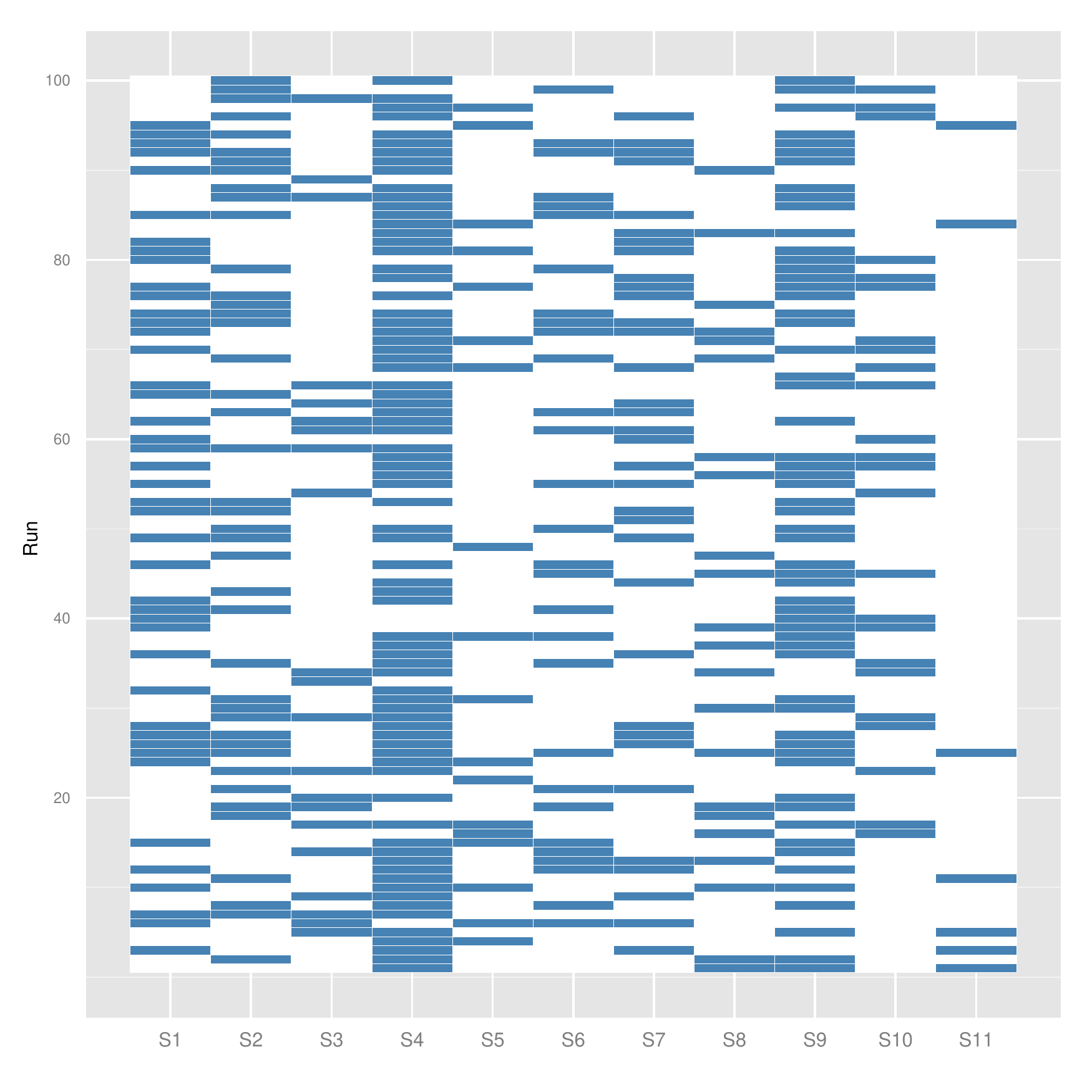}
\caption{ \label{fig:PopMs} Summary statistics selected in 100 runs of the automated summary statistic selection procedure on simulated data sets from our three population genetics models. (each run is performed on a different simulated observed data point). A) Statistics chosen for model selection with data generated from model $1$. B) Statistics chosen for model selection with data generated from model $2$. C) Statistics chosen for model selection with data generated from model $3$. }
\end{center}
\end{figure}

The results of the selection process when performed over 100 different simulated data sets from each of the three models considered are shown in figure \ref{fig:PopMs}. It is apparent that the chosen statistics vary between data sets generated quite considerably --- this is to be expected as the statistics required for sufficiency will vary depending on the data. For data generated by all of the models it is apparent that \textbf{S4}, the average SNP homozygosity, is selected often, whilst as expected the uninformative random statistic \textbf{S11} is rarely chosen. For data generated under the exponential growth model, \textbf{S4}, as well as \textbf{S3} (homozygosity), \textbf{S6} (mean number of pair-wise differences between haplotypes) and \textbf{S9} (linkage disequilibrium), appear to be favoured by the model selection approach. Data generated from this model apparently requires more statistics than data generated from the Null model to achieve sufficiency. The method applied to data generated from the two island model selects statistics \textbf{S4} and \textbf{S9} often, interestingly seeming to require fewer statistics than the exponential growth model to perform model selection.
\par
This is a new and initially perhaps surpising finding: the summaries chosen by our model selection approach depend subtly on the true data-generating model. This is, by hindsight, however, not unexpected: we are trying to achieve sufficiency for model parameters first, and then pool the statistics required to do just that for all models, before refining this set of statistics in order to obtain sufficiency for model selection. As some models will generate data that is more difficult to obtain under other models than is the case vice versa, such relative biases will affect the set of statistics chosen. In light of population genetics theory, therefore, our observations are completely in line with our understanding of coalescent processes (see \eg \cite{Hein:2005aa}). 

\subsection{Random walk models}

We also apply our framework to the problem of model selection on random walks, using a number of summary statistics. The models \citep{Rudnick:2010aa} under consideration were:

\begin{description}
\item[Model 1] Brownian motion.
\item[Model 2] Persistent random walk (where the walk is more likely to continue in the same direction over successive steps but does not have a particular favoured orientation).
\item[Model 3] Biased random walk (where one direction is favoured).
\end{description}

We used five summary statistics, that are individually not sufficient in more than one dimension for any of the random walk models:

\begin{description}
\item[S1] Mean square displacement.
\item[S2] Mean x and y displacement.
\item[S3] Mean square x and y displacement.
\item[S4] Straightness index.
\item[S5] Eigenvalues of gyration tensor (reference random walks book).
\end{description}

Since the models have multiple parameters, we can no longer apply a $\chi^2$ test to select sufficient statistics, and so instead we approximate the KL-divergence using the posterior, applying the formula described in \citep{Boltz:2007hg},
\begin{equation}
\label{eq:klc}
KL(p_X || p_Y) \approx \log\frac{N_V}{N_U-1}+d\mathrm{E}_U[\log \rho_k(\cdot,V)] - d\mathrm{E}_U[\log \rho_k(\cdot,U)],
\end{equation}
where $U$ and $V$ are the sets of posterior particles drawn from distributions $p_X$ and $p_Y$ respectively, $d$ is the number of parameters and $\mathrm{E}_U[\log \rho_k(\cdot,V)]$ is the expectation of the distance to the $k$th nearest neighbour in the set of particles $V$, $\rho_k(u,V)$ of each particle $u \in U$.

Applying formula \eqref{eq:klc} in our summary statistic selection framework to data simulated from the three different models over $100$ runs, the statistics shown in figure \ref{fig:RwMs} are chosen. Again it is apparent that there are some differences in the selected statistics for different data sets generated. Looking at the summary statistics selected by our method, statistic \textbf{S5}, the eigenvalues of the gyration tensor, a measure of the anisotropy of the random walk, appears to be chosen often for data generated by all three models. There also appears to be a slight preference for statistic \textbf{S2}, the mean x and y displacement, which can be understood given that this statistic is necessary for sufficiency for parameter inference on the biased random walk model. We need to stress that the structure of the data here is complex and summary statistics are expected to be hugely variable. 
\begin{figure}[p]
\begin{center}
\includegraphics[width=0.48\textwidth]{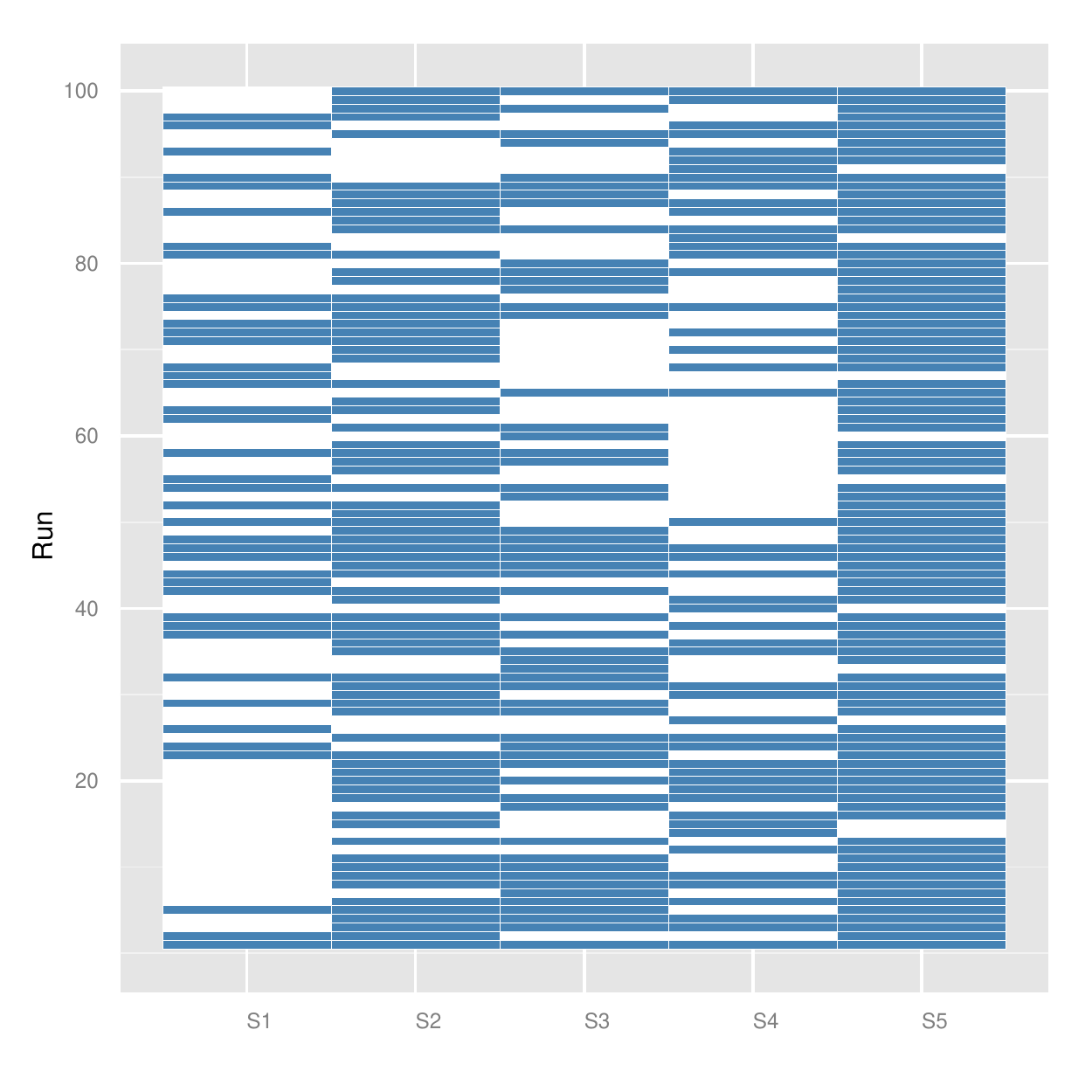}
\includegraphics[width=0.48\textwidth]{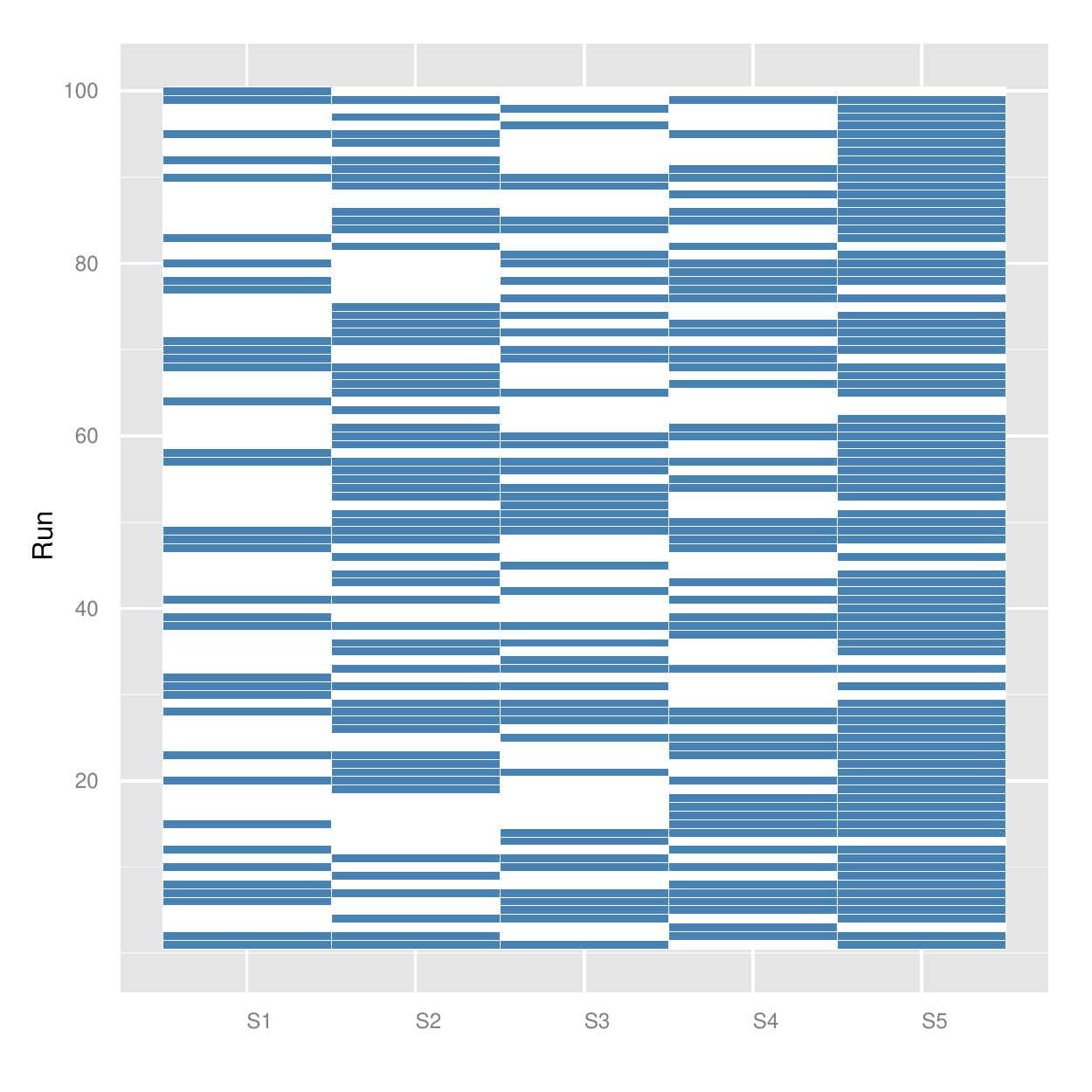}
\includegraphics[width=0.48\textwidth]{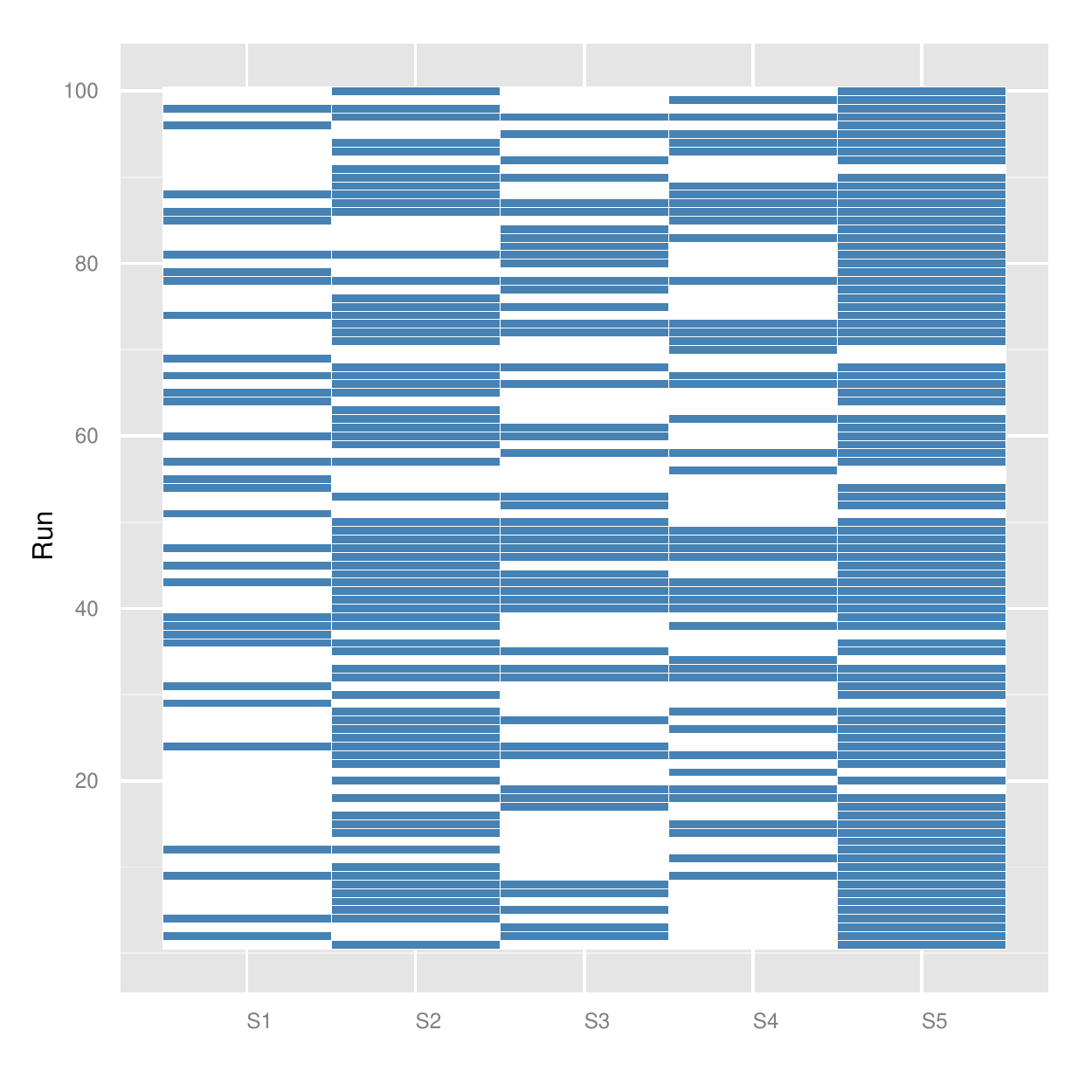}
\caption{ \label{fig:RwMs} Summary statistics selected in 100 runs of the automated summary statistic selection on different simulated data sets, considering three different random walk models. A) Statistics chosen for model selection with data generated from model $1$. B) Statistics chosen for model selection with data generated from model $2$. C) Statistics chosen for model selection with data generated from model $3$. }
\end{center}
\end{figure}

\section{Discussion}
Sufficient statistics are rare; in convenient form --- \ie where the number of statistics is equal to the number of parameters to be estimated --- they are restricted to problems that can be described in terms of models that belong to the exponential family \citep{Lehmann:1993aa,Didelot:2010wo}. As previous authors have pointed out it is necessary to develop methods that construct sets of statistics that are (at least approximately \citep{LeCam:1964wk,Kusama:1976uq}) sufficient \citep{Joyce:2008vm,Nunes:2010dv,Fearnhead:2010vj,Fearnhead:2010wg}. It is either this, or reinterpreting ABC-based inferences not as approximations to the full Bayesian (and thus likelihood-based) apparatus but as inference procedures in their own right \citep{Wilkinson:2008uy,Drovandi:2011cb}, potentially systematically biased or for approximate models. A third approach, previously advocated, is to consider model checking rather than model selection as a viable way of ensuring that only appropriate models are calibrated against data. We believe that the latter position fails to acknowledge the role of sufficiency of statistics also in the context of parameter estimation; and we will briefly return to this point below after having addressed the other two points.
\par
All methods aimed at constructing collections of statistics that taken together are (approximately) sufficient will fail, almost trivially, unless an exhaustive set of summary statistics can be envisioned which fulfils the sufficiency criteria as outlined above. If that is not the case, then we might naively expect that all candidate summary statistics from our starting set $\mathcal{S}$ will be included in $\U$. This, however, need not (and we believe generally will not) be the case, as the information theoretical framework will tend to bias against inclusion of statistics that are in some way co-linear to any statistics that are already included in the constructed set. It is, of course, in principle possible to use the KL divergence with respect to the distribution obtained with the full data as an overall benchmark, but in cases where this is indeed possible, it may be best to use the full data (see \eg \cite{Toni:2009gm}) for inference rather than risk the information reduction inherent to most summary statistics.
\par
There has been much interest in trying to interpret ABC not solely as an approximation to the ``true" posterior, but as an inferential framework in its own right \citep{Wilkinson:2008uy,Drovandi:2011cb}. This is perhaps an attractive option. One way of achieving this shift in perspective is to consider distributions such as
$$
p(M,\theta|S)
$$
as distributions which specify the probability of a parameter and model being in concordance with a given summary statistic. If all we care about is that a model and parameter combination have high (or low) probability of producing data with certain mean/maximum/minimum or any other summary statistic value, then this is perfect. It is easy to envisage scenarios where we are only interested in certain aspects of the data (such as maximum water levels). 
 ABC methods can be used to infer parameters (and models) that are more likely to give rise to simulated data that shares some but not all characteristics of the data. Interestingly, this would also allow us to employ ABC as a design tool \citep{Barnes:2011ua}: we specify the data (or system behaviour) that we would like to observe and infer parameters (and models) which have high probability of producing these types of behaviour.
\par
While this may perhaps seem like sophistry it does also have serious implications for model checking: any ABC approach that is based on summary statistics will infer model parameters (or marginal model posteriors) that reflect the behaviour encoded by these statistics. Thus we can no longer use these same statistics for model checking. This reflects the need to use non-sufficient summary statistics for model checking from Bayesian posterior predictive distributions: if we perform inference under a model for which a sufficient statistic exists, then calculating the same statistic for replicate data generated from the posterior predictive distribution will result in test statistics that are in line with the observed data, irrespective of the validity of the model. Hence some authors, in particular \cite{Gelman:2003}, strongly advocate the use of graphical model checking techniques over numerical tests. We feel that the situation in ABC reflects some of the same problems that are also encountered in model checking. Thus in an ABC framework, irrespective of whether the statistics are sufficient or not, the posterior distributions reflect the choice of statistics and the same statistics are therefore ill-suited for model checking.
\par
We conclude by reiterating that ABC approaches employing summary statistics rather than the whole data have to fully engage with the level of information-loss inherent to summary statistics. Notions of simple sufficient statistics probably do not apply for most scientifically interesting and challenging problems and the use of statistics rather than the real data will always result in loss of information. Our approach is based around the assessment of information loss and allows the principled construction of sets of statistics (from a candidate set) that capture as much as possible from the observed data. While not a panacea, it is within the computational reach of ABC practitioners and makes information loss due to inadequate use of statistics apparent, for both the parameter and model selection problems.

\bibliographystyle{interface}
\bibliography{bibsuffstat.bib,/Users/michael/bibliography/mstbibnet.bib,/Users/michael/bibliography/wholebib.bib}

\end{document}